\documentclass[runningheads]{llncs}
\usepackage[T1]{fontenc}
\usepackage{graphicx}
\usepackage[nothing]{algorithm}
\usepackage{algorithmicx}
\usepackage{setspace}
\usepackage[noend]{algpseudocode}
\usepackage{verbatim}
\usepackage{amsmath}
\usepackage{cleveref}
 \usepackage{amssymb}

\usepackage{enumerate}

\begin{document}
\title{Fully Dynamic $k$-Center in Low Dimensions via Approximate Furthest Neighbors}
%
%
\author{Jinxiang Gan \and
Mordecai Jay Golin}
\authorrunning{J.Gan and M.J.Golin}
%
\institute{Hong Kong University of Science and Technology \\ \email{\{jganad,golin\}@cse.ust.hk}}
\maketitle              
\begin{abstract}
Let $P$  be a set of points in  some metric space.
  The approximate furthest neighbor problem is, given a second point set $C,$ to find a point $p \in P$ that is a $(1+\epsilon)$ approximate furthest neighbor from $C.$
  
  The dynamic version  is to maintain $P,$ over insertions and deletions of points, in a way that permits efficiently solving the approximate furthest neighbor problem for the current $P.$ 
  
  We provide the first algorithm for  solving this problem  in metric spaces with finite doubling dimension. Our algorithm is built on top of the navigating net data-structure.

An immediate  application is two   new algorithms for solving the dynamic $k$-center problem.  The first dynamically maintains  $(2+\epsilon)$ approximate  $k$-centers in general metric spaces with bounded doubling dimension and  the second maintains 
    $(1+\epsilon)$ approximate  Euclidean $k$-centers. Both these dynamic algorithms work by starting with a known corresponding static algorithm for solving approximate $k$-center,  and replacing the static exact furthest neighbor subroutine used by that algorithm  with  our new dynamic approximate furthest neighbor one.

 Unlike previous algorithms for dynamic $k$-center
    with those same approximation ratios, our new ones do not require knowing $k$ or $\epsilon$ in advance. In the Euclidean case, our algorithm also seems to be the first deterministic solution.

\end{abstract}

\section{Introduction}
    The main technical result of this paper is an efficient procedure for calculating approximate furthest neighbors from a dynamically changing point set $P.$ This procedure, in turn,  will lead to the development of two new simple algorithms for maintaining approximate $k$-centers in dynamically changing point sets.

    Let $B(c,r)$ denote the ball centered at $c$ with radius $r$. The $k$-center problem is to find a minimum radius $r^*$ and associated $C$ such that the union of balls $\bigcup_{c \in C} B(c,r^*)$ contains all of the points in $P.$
     
    In the arbitrary metric space version  of the problem, the centers are restricted to be points in $P.$  In the {\em Euclidean} $k$-center problem  $(\mathcal{X},d)=\left(\mathbb{R}^D,\ell_2\right)$ with $D \ge 1$ and $C$ may be any  set of $k$ points in $\mathbb{R}^D.$ The Euclidean $1$-center problem  is also known as the minimum enclosing ball (MEB) problem. 
          
    An $\rho$-approximation algorithm would find a set of centers $C',$ $|C'|\le k,$ and radius $r'$ in polynomial time such that
    $\bigcup_{c \in C'} B(c,r')$ contains all of the points in $P$ and $r' \le \rho r^*$. The $k$-center problem is known to be NP-hard to approximate with a factor smaller than 2 for  arbitrary metric spaces\cite{hsu1979easy}, and with a factor smaller than $\sqrt{3}$ for  Euclidean spaces \cite{feder1988optimal}.

    \paragraph*{Static algorithms}
    There do exist two 2-approximation algorithms in \cite{gonzalez1985clustering,hochbaum1985best} for the $k$ center problem on an arbitrary metric space; the best-known approximation factor for Euclidean $k$-center remains 2 even for two-dimensional space when $k$ is part of the input (see \cite{feder1988optimal}). There are better results for  the special case of the Euclidean $k$-center for fixed $k$,  $k=1$ or $2$ (e.g., see \cite{badoiu2002approximate,badoiu2003smaller,agarwal2010streaming,kim2015improved,kim1}). There are also PTASs \cite{badoiu2002approximate,badoiu2003smaller,kim1} 
    for the Euclidean $k$ center when $k$ and $D$ are  constants.

    \paragraph*{Dynamic algorithms}

    In many practical applications, the data set $P$ is not static but changes {\em dynamically} over time, e.g, a new point may be inserted to or deleted from $P$ at each step. 
    $C$ and $r$  then need to be recomputed at selected query times.   If only insertions are permitted, the problem is {\em incremental};
    if both insertions and deletions are permitted, the problem is  {\em fully dynamic}.
    
    The running time of such dynamic algorithms are often  split into the time required for an {\em update} (to register a change in the storing data structure) and the time required for a {\em query} (to solve the problem on the current dataset). In dynamic algorithms, we require  both update and query time to be  {\em nearly logarithmic} or constant. The static versions   take linear time.

    Some known results on these problems are listed in \Cref{table:results}.
    As is standard, many of them are stated in terms of the {\em aspect ratio} of point set $P$. Let $d_{max}=\sup\{d(x,y):x,y\in P \text{ and } x\neq y\}$ and  $d_{min}=\inf\{d(x,y):x,y\in P \text{ and } x\neq y\}$. The  {\em aspect ratio} $\Delta$ of $P$ is $\Delta=\frac{d_{max}}{d_{min}}$.

    \begin{table}
     \setlength{\leftskip}{40pt}
    \scalebox{0.9}{
    \begin{tabular}{|cccccc|}
    \hline
     \multicolumn{6}{|c|}{Arbitrary Metric Space $(\mathcal{X},d)$} \\ \hline

    \multicolumn{1}{|c|}{Author} & \multicolumn{1}{c|}{Approx.} & \multicolumn{1}{c|}{Dimensions} & \multicolumn{1}{c|}{Update Time}   &   \multicolumn{1}{c|}{Query Time}     &
    \multicolumn{1}{c|}{Fixed}\\ \hline
    
    \multicolumn{1}{|c|}{Chan et al. \cite{chan2018fully}}       & \multicolumn{1}{c|}{$2+\epsilon$}        & \multicolumn{1}{c|}{High}           & \multicolumn{1}{c|}{O($k^2\frac{\log \Delta}{\epsilon}$) (avg.)}   &        \multicolumn{1}{c|}{$O(k)$}   &   \multicolumn{1}{c|}{$k,\epsilon$}                   \\ \hline
    
    \multicolumn{1}{|c|}{ Goranci et al. \cite{goranci2021fully} }       & \multicolumn{1}{c|}{$2+\epsilon$}        & \multicolumn{1}{c|}{Low}           & \multicolumn{1}{c|}{$O((2/\epsilon)^{O(dim(\mathcal{X}))} \log \Delta \log \log \Delta \cdot \ln \epsilon^{-1})$}               &        \multicolumn{1}{c|}{$O( \log \Delta + k)$}    &   \multicolumn{1}{c|}{$\epsilon$}     \\ \hline
    
    \multicolumn{1}{|c|}{Bateni et al. \cite{bateni2021optimal}}       & \multicolumn{1}{c|}{$2+\epsilon$}        & \multicolumn{1}{c|}{High}           & \multicolumn{1}{c|}{$O(\frac{\log \Delta \log n}{\epsilon}(k+  \log n)$)(avg.)}    &           \multicolumn{1}{c|}{$O(k)$}     &    \multicolumn{1}{c|}{$k,\epsilon$}               \\ \hline
   \multicolumn{1}{|c|}{This paper}       & \multicolumn{1}{c|}{$2+\epsilon$}        & \multicolumn{1}{c|}{Low}           & \multicolumn{1}{c|}{$O\left(2^{O(dim(\mathcal{X}))}\log \Delta\log\log \Delta\right)$}               &          \multicolumn{1}{c|}{$O( k^2(\log \Delta+(1/\epsilon)^{O(dim(\mathcal{X})}))$}     &    \\ \hline
    
    \end{tabular}
    }
    \ \\\     
    \setlength{\leftskip}{-90pt}
    \scalebox{0.9}{
   \begin{tabular}{|cccccc|}
    \hline 
    \multicolumn{6}{|c|}{Euclidean Space $(\mathbb{R}^D,\ell_2)$} \\ \hline
    \multicolumn{1}{|c|}{Author} & \multicolumn{1}{c|}{Approx.} & \multicolumn{1}{c|}{Dimensions} & \multicolumn{1}{c|}{Update Time}   &   \multicolumn{1}{c|}{Query Time}    &
    \multicolumn{1}{c|}{Fixed}\\ \hline

     \multicolumn{1}{|c|}{Chan \cite{chan2009dynamic}}       & \multicolumn{1}{c|}{$1+\epsilon$}        & \multicolumn{1}{c|}{Low}           & \multicolumn{1}{c|}{$O((\frac{1}\epsilon)^Dk^{O(1)}\log n)$ (avg)}    &    \multicolumn{1}{c|}{$O(\epsilon^{-D}k\log k \log n+(\frac{k}\epsilon)^{O(k^{1-1/D})})$}                        &  \multicolumn{1}{c|}{$k,\epsilon$}     \\ \hline 
     
   \multicolumn{1}{|c|}{Schmidt and Sohler \cite{schmidt2019fully} }       & \multicolumn{1}{c|}{$16$}        & \multicolumn{1}{c|}{Low}           & \multicolumn{1}{c|}{$O((2\sqrt{d}+1)^d\log^2 \Delta \log n)$ (avg.)}    &           \multicolumn{1}{c|}{$O((2\sqrt{d}+1)^d(\log \Delta+ \log n))$}            &     \multicolumn{1}{c|}{}     \\ \hline 
          
     \multicolumn{1}{|c|}{Schmidt and Sohler \cite{schmidt2019fully} }       & \multicolumn{1}{c|}{$O(f\cdot D)$}        & \multicolumn{1}{c|}{High}           & \multicolumn{1}{c|}{$O(D^2\log^2n \log \Delta n^{1/f})$ (avg.)}    &           \multicolumn{1}{c|}{$O(f\cdot D\cdot\log n\log \Delta)$}            &     \multicolumn{1}{c|}{$f$}     \\ \hline 
    
    \multicolumn{1}{|c|}{(*) Bateni et al. \cite{bateni2021optimal}}       & \multicolumn{1}{c|}{$f(\sqrt{8}+\epsilon)$}        & \multicolumn{1}{c|}{High}           & \multicolumn{2}{c|}{$O(\frac{\log \delta^{-1} \log \Delta}{\epsilon}Dn^{1/f^{2}+o(1)})$}             &   \multicolumn{1}{c|}{$\epsilon,f$}     \\ \hline
       \multicolumn{1}{|c|}{This paper}       & \multicolumn{1}{c|}{$1+\epsilon$}        & \multicolumn{1}{c|}{Low}           & \multicolumn{1}{c|}{$O\left(2^{O(D)}\log \Delta\log\log \Delta\right)$}               &          \multicolumn{1}{c|}{$O(D\cdot k(\log \Delta+(1/\epsilon)^{O(D)})2^{k\log k/\epsilon})$}     &    \\ \hline
    \end{tabular}}
        \caption{
        Previous results on approximate dynamic $k$-centers.  More information on the model used by each is in the text. Note that all algorithms listed provide correct results except for Schmidt and Sohler \cite{schmidt2019fully}, which maintains  $O(f\cdot D)$  with probability $1-1/n$, and Bateni et al. \cite{bateni2021optimal}, which  maintains a $f(\sqrt{8}+\epsilon)$ solution with  probability $1-\delta$. \cite{bateni2021optimal} also combines the updates and queries.
            }
            \label{table:results}
    \end{table}   The algorithms  listed in the table work under slightly different models.  More explicitly:
    \begin{enumerate}
        \item  For arbitrary  metric spaces, both \cite{goranci2021fully} and the current paper assume that the metric space has a bounded doubling dimension $dim(\mathcal{X})$ (see \Cref {def:doubling}).
        \item In ``Low dimension'',  update time may be  exponential in $D$; in  ``High dimension'' it may not.
        \item The ``fixed'' column denotes parameter(s) that must be fixed in advance when  initializing  the corresponding data structure, e.g., $k$ and/or $\epsilon.$ In addition, in both  
        \cite{schmidt2019fully,bateni2021optimal} for high dimensional space,
        $f\geqslant 1$ is a constant selected in advance that appears in both the approximation factor and running time. 
        
        \smallskip
        The  data structure used in the current paper is the navigating nets from \cite{krauthgamer2004navigating}. It does not require knowing $k$ or $\epsilon$ in advance but instead   supports them as parameters to the  query.
       \item In \cite{chan2009dynamic}, (avg) denotes that the update time is in  expectation (it is a randomized algorithm).
        \item  
        Schmidt and Sohler \cite{schmidt2019fully}  answers the  slightly different
         {\em membership query}.  Given $p$, it returns the cluster containing $p.$ 
        In low dimension, the running time of their algorithm is expected and amortized.
    \end{enumerate}

    \paragraph*{Our contributions and techniques}
    Our main results are two  algorithms for solving the dynamic  approximate $k$-center problem in, respectively,  arbitrary  metric spaces with a finite doubling dimension and in  Euclidean space. 

    \begin{enumerate}
        \item    Our first new algorithm is for {\em any metric space with finite doubling dimension}: 
        \begin{theorem} \label{thm:main2}
        Let $(\mathcal{X},d)$ be a metric space with a finite doubling dimension $D$. Let $P \subset X$ be  a dynamically changing set of points. We can  maintain $P$ in $O(2^{O(D)}\log \Delta\log \log \Delta)$ time per point insertion and deletion so as to support 
        $(2 + \epsilon)$ approximate $k$-center
        queries in 
        $O(k^2(\log \Delta+(1/\epsilon)^{O(D)}))$ time.
        \end{theorem}
        Compared with previous results (see table \ref{table:results}), our data structure does not require knowing $\epsilon$ or $k$ in advance, while the construction of the previous data structure depends on $k$ or $\epsilon$ as basic knowledge.
        \item     Our second new algorithm is for the Euclidean $k$-center problem:
        \begin{theorem} \label{thm:main}
        Let $P \subset \mathbb{R}^D$ be  a dynamically changing set of points. We can maintain $P$ in $O(2^{O(D)}\log \Delta\log \log \Delta)$ time per point insertion and deletion so as to support 
        $(1 + \epsilon)$ approximate $k$-center
        queries 
        in 
        $O(D\cdot k(\log \Delta+(1/\epsilon)^{O(D)})2^{k\log k/\epsilon})$ time.
        \end{theorem}
    This algorithm seems to be the first deterministic dynamic solution for the Euclidean $k$-center problem. Chan \cite{chan2009dynamic} presents a randomized dynamic algorithm while they do not find a way to derandomize it.
    \end{enumerate}
    
    The motivation for our new approach was the observation that many previous results e.g., \cite{badoiu2003smaller,badoiu2002approximate,chan2009dynamic,gonzalez1985clustering,kim1}, on  static $k$-center, work by iteratively searching the furthest neighbor in $P$ from a changing  set of points $C.$

    The main technical result of this paper is an efficient procedure for calculating {\em approximate} furthest neighbors from a dynamically changing point set $P.$ This procedure, in turn,  will lead to the development of two new simple algorithms for maintaining approximate $k$-centers in dynamically changing point sets.

    Consider a set of $n$ points $P$ in some metric space $({\mathcal X},d).$  A nearest neighbor in  $P$  to a query point $q,$ is a point $p \in P$ satisfying
    $d(p,q) = \min_{p'\in P}d(p',q) = d(P,q).$ A {\em $(1 + \epsilon)$ approximate nearest neighbor to $q$} is a point $p \in P$ satisfying $d(p,q) \le (1+\epsilon) d(P,q).$

    Similarly,  a furthest neighbor  to a query point $q$ is a $p$ satisfying $d(p,q) = \max_{p'\in P}d(p',q)$.
    A  {\em $(1 + \epsilon)$ approximate furthest  neighbor to $q$ } is a point $p \in P$ satisfying $\max_{p' \in P}d(p',q) \le (1+\epsilon)d(p,q).$
    
    There exist efficient algorithms for maintaining a {\em dynamic} point set $P$ (under insertions and deletions) that, given  query  {\em point} $q$,  quickly permit calculating approximate nearest \cite{krauthgamer2004navigating} and furthest \cite{bespamyatnikh1996dynamic,pagh2015approximate,chan2016dynamic} neighbors to $q.$
    
    A {\em $(1 + \epsilon)$   approximate nearest neighbor} to a query {\em set} $C$, is a point $p \in P$ satisfying $d(p,C) \le (1+\epsilon) d(P,C)$.  Because ``nearest neighbor''  is decomposable, i.e., 
    $d(P,C) = \min_{q \in C} d(P,q),$ 
    \cite{krauthgamer2004navigating} also permits efficiently calculating an approximate nearest neighbor to set $C$ from a dynamically changing $P.$
    
    An approximate furthest  neighbor to a query {\em set} $C$ is similarly defined as a point $p \in P$ satisfying $\max_{p' \in P} d(p',C) \le (1+\epsilon) d(p,C).$  
    Our main new technical result is \Cref {thm:NNAFN}, which permits efficiently calculating an approximate  {furthest}  neighbor to query set $C$ from a dynamically changing $P.$
    We note that, unlike nearest neighbor, furthest neighbor is not a decomposable problem and such a procedure does not seem to have previously known.

    This technical result  permits the creation of new algorithms for solving the dynamic {\em $k$-center problem} in low dimensions.

\section{Searching for a $(1+\epsilon)$-Approximate Furthest Point in a Dynamically Changing Point Set}
\label{Sec:AFNC}
Let $(\mathcal{X},d)$ denote a fixed  metric space.
     \begin{definition}
     Let  $C,P \subset \mathcal{X}$ be finite sets of points  and
     $q \in \mathcal{X}$. Set 
     $$d( {C},q)= d(q,C) = \min_{q'\in{C}}d(q',q)
     \quad\mbox{and}\quad
     d({C},P)=\min_{p\in P}d(C,p).
     $$
     
     
            
            
  
       
            $p \in P$ is a   {\em furthest neighbor in $P$ to  $q$}  if 
            $d(q,p)= \max_{p' \in P}d(q,p').$ 
            
                     $p \in P$ is a    {\em  furthest neighbor in $P$ to set $C$}    if 
            $d(C,p)= \max_{p' \in P}d(C,p').$
            
            $p \in P$ is a    {\em $(1+\epsilon)$-approximate furthest neighbor in $P$ to $q$}   if
            $$\max_{p' \in P}d(q,p') \le (1+\epsilon)d(q,p).$$
%
%

            $p \in P$ is a    {\em $(1+\epsilon)$-approximate furthest neighbor in $P$ to $C$}   if
            $$\max_{p' \in P}d(C,p') \le (1+\epsilon)d(C,p).$$

            $\mbox{FN}(P,q)$ and 
             $\mbox{AFN}(P,q,\epsilon)$ will, respectively, denote procedures returning a furthest neighbor and a $(1+\epsilon)$-approximate furthest neighbor to $q$ in $P.$
            
                        $\mbox{FN}(P,C)$ and 
             $\mbox{AFN}(P,C,\epsilon)$ will, respectively, denote procedures returning a furthest neighbor and $(1+\epsilon)$-approximate furthest neighbor to $C$ in $P.$
%
%
%
%
%
    \end{definition}

   Our algorithm assumes that $\mathcal{X}$  has  finite doubling dimension.
  
   \begin{definition}[Doubling Dimensions]\label{def:doubling}
    The doubling dimension of a metric space $(\mathcal{X},d)$ is the minimum value $\dim(\mathcal{X})$ such that any ball $B(x,r)$ in $(\mathcal{X},d)$ can be covered by $2^{\dim(\mathcal{X})}$ balls of radius $r/2$. 
    \end{definition}
        
    It is known that the doubling dimension of the Euclidean space $(R^D,\ell_2)$ is $\Theta(D)$ \cite{heinonen2001lectures}. 
    
        
   Now let $(\mathcal{X},d)$ be  a metric space with  a finite doubling dimension  and 
   $P \subset \mathcal{X}$ be  a  finite set of points.
   Recall that 
   $d_{max}=\sup\{d(x,y):x,y\in P\}
   \quad\mbox{and}\quad
   d_{min}=\inf\{d(x,y):x,y\in P,\ x\neq y\}.$ and 
   The  {\em aspect ratio} $\Delta$ of $P$ is $\Delta=\frac{d_{max}}{d_{min}}$.  
    
    
    Our main technical theorem (proven below in \Cref{subsec:AFN}) is:
    \begin{theorem} \label{thm:NNAFN}
    Let $(\mathcal{X},d)$ be a metric space with finite doubling dimension and $P \subset \mathcal{X}$ be a point set stored by a 
        navigating net data structure \cite{krauthgamer2004navigating}.
        Let $C \subset \mathcal{X}$ be another point set.
        Then, 
        we can find a $(1+\epsilon)$-approximate furthest point among $P$ to  $C$ in $O\left(\mathcal{|C|}(\log \Delta+(1/\epsilon)^{O(\dim(\mathcal{X}))})\right)$ time, where $\Delta$ is the aspect ratio of set $P$.  
    \end{theorem}

The {\em navigating net} data structure \cite{krauthgamer2004navigating} is described in more detail below. 
    \subsection{Navigating Nets \cite{krauthgamer2004navigating}}
   Navigating nets are very well-known structures for dynamically maintaining points in a metric space with finite doubling dimension, in a way that  permits approximate nearness queries.  To the best of our knowledge they have not been previously used for approximate ``furthest point from set'' queries.
    
  To describe the algorithm, we first need to quickly review some basic known facts about navigating nets.
        The following  lemma is critical to our analysis.
        \begin{lemma}\cite{krauthgamer2004navigating}\label{lem:doubling}
            Let $(\mathcal{X},d)$ be a metric space and $Y\subseteq \mathcal{X}$. If the aspect ratio of the metric induced on $Y$ is at most $\Delta$ and $\Delta\geqslant 2$, then $|Y|\leqslant \Delta^{O(\dim(\mathcal{X}))}$.
        \end{lemma}

        We next introduce some  notation from \cite{krauthgamer2004navigating}: 
        \begin{definition}[$r$-net] \cite{krauthgamer2004navigating}
            Let $(\mathcal{X},d)$ be a metric space. For a given parameter $r>0$, a subset $Y\subseteq \mathcal{X}$ is an $r$-net of $P$ if it satisfies:
           \begin{enumerate}[(1)]
                \item For every $x,y \in Y$, $d(x,y)\geqslant r$;
                \item $\forall x\in P$, there exists at least one $y\in Y$ such that $x\in B(y,r)$. 
            \end{enumerate}
        \end{definition}
        
        We now start the description of the 
        navigating net  data structure. Set $\Gamma=\{2^i:i\in \mathbb{Z}\}$. Each $r\in \Gamma$ is called a {\em scale}. 
        For every $r\in \Gamma$,  $Y_r$ will denote  an $r$-net of $Y_{r/2}$. The base case is that for every scale $ r\leqslant d_{min}$, $Y_r=P.$ 
        

        Let $\gamma\geqslant 4$
         be  some fixed  constant.  
        For each scale $r$ and each $y\in Y_r$, the data structure stores the set of points
        \begin{equation}
        \label{eq:L}
            L_{y,r}=\{z\in Y_{r/2}: d(z,y)\leqslant \gamma\cdot r\}.
        \end{equation}
        $L_{y,r}$ is called the {\em scale $r$ navigation list of $y$}.
        
        Let $r_{max} \in \Gamma$ denote the smallest $r$ satisfying $|Y_r|=1$ and $r_{min}\in \Gamma$ denote the largest $r$ satisfying $L_{y,r}=\{y\}$ for every $y\in Y_r$. Scales $r\in [r_{min},r_{max}]$ are called {\em non-trivial} scales; all other scales are called  {\em trivial}. Since $r_{max}=\Theta(d_{max})$ and $r_{min}=\Theta (d_{min})$, the number of non-trivial scales is
         $O\left( \log_2 \frac {r_{max}} {r_{min}} \right) = O( \log_2 \Delta).$
%
        
        
       Finally, we need a few more basic properties of  navigating nets:
        \begin{lemma}\cite{krauthgamer2004navigating}(Lemma 2.1 and 2.2)\label{lem:property_of_nn}
        For each scale $r$, we have: 
        \begin{enumerate}[(1)]
        \item $\forall y \in Y_r$,
           $|L_{y,r}|=O(2^{O(\dim(\mathcal{X}))}).$
            \item $\forall z\in P$, $d(z, Y_r)<2r$;
            \item $\forall x,y \in Y_r$, $d(x,y)\geqslant r$.
        \end{enumerate} 
        \end{lemma}

        
            We provide an example (\Cref{fig:NN}) of navigating nets in the Appendix. Navigating nets were originally designed to solve dynamic approximate nearest neighbor queries and are useful because they can be quickly updated.

        \begin{theorem} (\cite{krauthgamer2004navigating})
           Navigating nets
           use $O(2^{O(\dim(\mathcal{X}))}\cdot n)$ words. The data structure can be updated with an insertion of a point to $P$ or a deletion of a point in $P$ in time $(2^{O(\dim(\mathcal{X}))}\log \Delta\log\log \Delta)$      \footnote{\em  Note: Although the update time of the navigating net depends on $O(\log \Delta)$, it does not explicitly maintain the value of $\Delta.$ Instead  it  dynamically maintains the values $r_{max}=\Theta(d_{max})$ and $r_{min}=\Theta(d_{min})$. The update time depends on the number of non-trivial scales $\log \frac{r_{max}}{r_{min}}=\Theta(\log \Delta)$, but without actually knowing $\Delta.$}. This includes $(2^{O(\dim(\mathcal{X}))}\log \Delta)$ distance computations.
        \end{theorem}

        

    \subsection{The Approximate Furthest Neighbor Algorithm $\mbox{AFN}(P,C,\epsilon$)}
    \label{subsec:AFN}
    
           \begin{algorithm}[h!]
        \caption{Approximate Furthest Neighbor: $\mbox{AFN}(P,C, \epsilon)$} 
        \label{alg:AFNC2}
        {\bf Input:} 
        A navigating net for set $P \subset \mathcal{X}$, set $C\subset \mathcal{X}$ and a constant $\epsilon>0$.\\
        {\bf Output:} A $(1+\epsilon)$-approximate furthest neighbor among $P$ to $C$ 
        \begin{algorithmic}[1]
            \State	Set $r=r_{max}$ and $Z_r=Y_{r_{max}}$;	
            \While{$r>\max \{\frac{1}{2}(\epsilon \cdot \max_{z\in Z_r}d(z,C)),r_{min}$\} }
            \State set 
            $Z_{r/2}=\bigcup_{z\in Z_r}\{y\in L_{z,r}: d(y,C)\geqslant \max_{z\in Z_r}d(z,C)-r\}$;
            \State set $r=r/2$
            \EndWhile
            \State Return $z\in Z_r$ satisfying $d(z,C)$ is maximal.
        \end{algorithmic}
        \end{algorithm}

     $\mbox{AFN}(P,C,\epsilon)$ is given in  \Cref {alg:AFNC2}. \Cref {fig:correctness_of_AFN} provides some geometric intuition.   $\mbox{AFN}(P,C,\epsilon)$
     requires that $P$ be stored in a navigating net  and the
    following definitions: \begin{definition} [The sets $Z_r$] \label{def:Z}\ 
        \begin{itemize}
        \item  $Z_{r_{max}}=Y_{r_{max}}$, where $|Y_{r_{max}}|=1$;
        \item If $Z_r$ is defined, 
         $Z_{r/2}=\bigcup_{z\in Z_r}\{y\in L_{z,r}: d(y,C)\geqslant \max_{z\in Z_r}d(z,C)-r\}$.
        \end{itemize}
        \end{definition}
         Note that, by induction,
        $Z_r \subseteq Y_r.$

        
        
        
            

        \begin{figure}[t]
        \centering
        \includegraphics[scale=0.3]{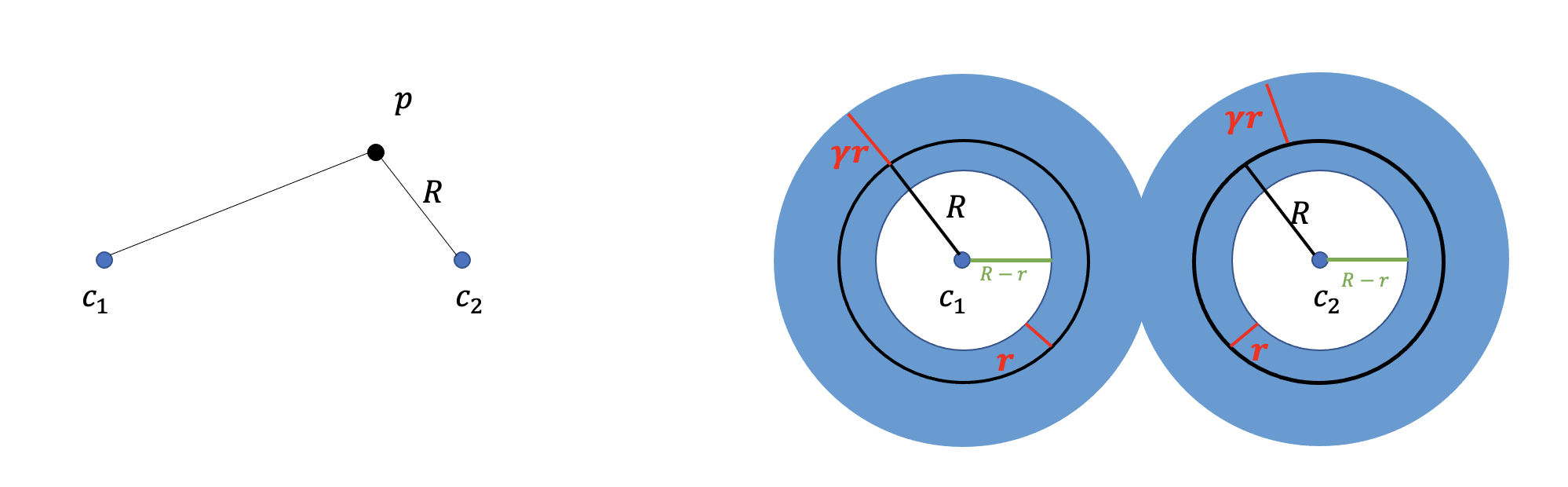}
        \caption{Illustration of line 3 of \Cref{alg:AFNC2}. $C=\{c_1,c_2\}$. Let  $p\in Z_r$ be  the furthest point in $Z_r$  to $C$ and set  $R=\max_{z\in Z_r}d(z,C)$ (in \cref{fig:correctness_of_AFN}, $R= d(c_2,p)).$, $Z_r\subseteq B(c_1,R)\cup B(c_2,R)$. If $z \in  B(c_i,R)$ and $y \in L_{z,r}$ then $y \in B(c_i,R + \gamma r).$     
        Next note that if $y \in Z_{r/2}$ then $y \in L_{z,r}$ for some $z \in Z_r$ and  for each $i,$
         $d(y,c_i)   \ge d(y,C)  \ge R -r.$
        This is illustrated in the right figure; $y$ must be in one of the two blue annulus  $B(c_i,R + \gamma r) \setminus B(c_i,R-r)$ ($i=1,2).$ Thus  $Z_{r/2}$ is contained in the union of the  annulus.}
        \label{fig:correctness_of_AFN}
        \end{figure}

        We now prove that $\mbox{AFN}(P,C,\epsilon)$ returns a $(1+\epsilon)$-approximate furthest point among $P$ to $C$. We start by showing  that, for every scale $r,$  the furthest point to $C$ is close to $Z_r.$
        
        
        \begin{lemma}\label{lem:proof1}
        Let $a^*$ be the furthest point to $C$ in $P$. Then, every set $Z_r$ 
        as defined in Definition \ref{def:Z},
        contains a point $z_r$ satisfying $d(z_r,a^*)\leqslant 2r$ \label{the furthest point is close to Z_r}
        \end{lemma}
        \begin{proof}

      The proof is illustrated in \Cref{fig:correctness_of_AFN2}. It works by downward induction on $r$. In  the base case $r = r_{max}$ and $Z_{r_{max}} = Y_{r_{max}}$, thus $d(a^*,Z_{r_{max}})\leqslant 2r$.
        
        For the inductive step, we  assume that $Z_r$ satisfies the induction hypothesis, i.e, $Z_r$ contains a point $z'$ satisfying $d(z',a^*)\leqslant 2r$. We will show that $Z_{r/2}$ contains a point $y$ satisfying $d(y,a^*)\leqslant r$.
        
        Since $Y_{r/2}$ is a $\frac{r}{2}$-net of $P$, there exists a point $y\in Y_{r/2}$ satisfying $d(y,a^*)\leqslant r$ (Lemma \ref{lem:property_of_nn}(2)).
        Then, $$d(z',y)\leqslant d(z',a^*)+d(a^*,y)\leqslant 2r+r=3r$$ and thus, because $\gamma \geqslant 4,$  $y\in L_{z',r}.$ Finally, let  $c'=\arg\min_{c_i\in C} d(y,c_i)$. Then $$d(y,C)=d(y,c')\geqslant d(a^*,c')-d(a^*,y)\geqslant d(a^*,C)-d(a^*,y)\geqslant \max_{z\in Z_r}d(z,C)-r.$$ Thus $y\in Z_{r/2}$.
        \end{proof}

         \begin{figure}[t]
        \centering
        \includegraphics[scale=0.22]{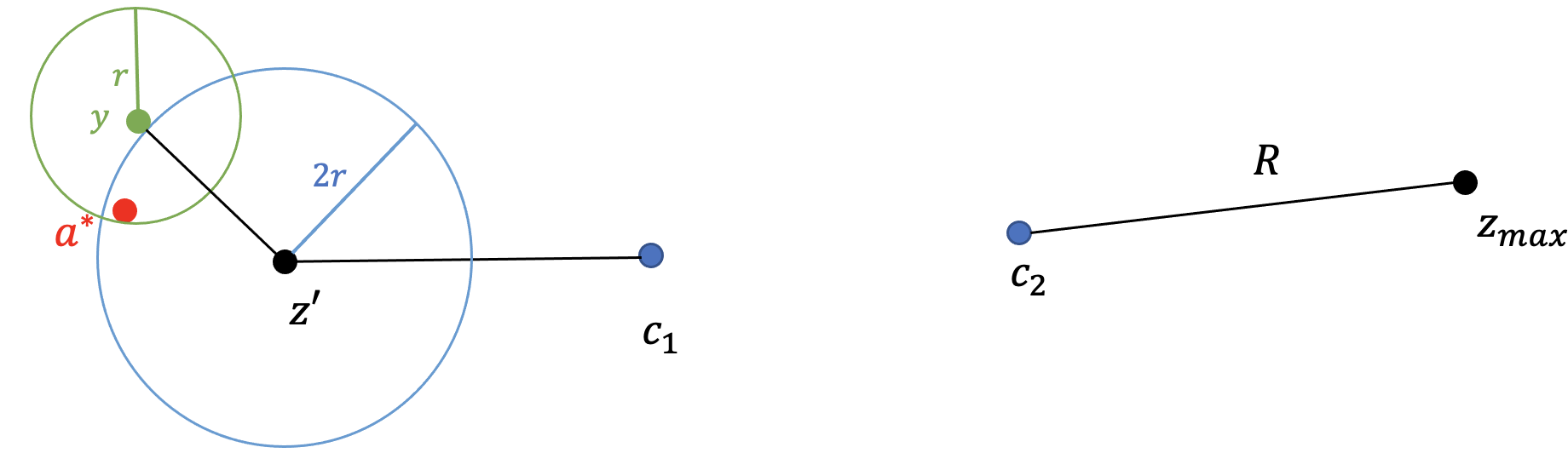}
        \caption{Illustration of  \Cref{lem:proof1}, for an example in which $C=\{c_1,c_2\}$ and $d(c_2,z_{max})=\max_{z\in Z_r}d(z,C)$. Suppose that, at scale $r$, $Z_r$ contains a point $z'$ satisfying $d(a^*,z')\leqslant 2r$. The proof shows  that $Z_{r/2}$ contains a point $y$ satisfying $d(y,a^*)\leqslant 2\cdot (r/2)=r$.}
        \label{fig:correctness_of_AFN2}
        \end{figure}
        
        Lemma \ref{lem:proof1} permits bounding the approximation ratio of algorithm $\mbox{AFN}(P,C,\epsilon)$.
        \begin{lemma}\label{lem:bound2}
        Algorithm $\mbox{AFN}(P,C,\epsilon)$ returns a point $q$ whose distance to $C$ satisfies $\max_{p\in P}d(p,C)\leqslant (1+\epsilon)d(q,C)$.
        \end{lemma}
        \begin{proof}
        Let $r'$ denote the value of $r$ at the end of the algorithm. Let $a^*$ be the furthest point to $C$ among $P$.
        Consider the two following conditions on $r':$
        \begin{enumerate}
            \item $r'\leqslant \frac{1}{2}(\epsilon\cdot \max_{z\in Z_{r'}} d(z,C))$. In this case, by Lemma \ref{the furthest point is close to Z_r}, there exists a point $z_{r'}\in Z_{r'}$ satisfying $d(z_{r'},a^*)\leqslant 2r'$. Let  $c'=\arg\min_{c_i\in C} d(z_{r'},c_i)$. 
        $$
        \begin{aligned}
            \max_{z\in Z_{r'}} d(z,C)&\geqslant d(z_{r'},C)=d(z_{r'},c')\geqslant d(a^*,c')-d(z_{r'},a^*)\\
           &\geqslant d(a^*,C)-2r'\geqslant d(a^*,C)-\epsilon \cdot \max_{z\in Z_{r'}} d(z,C)
        \end{aligned} $$  
        Thus, 
        \begin{equation}\label{eq:point 1}
            (1+\epsilon)\cdot \max_{z\in Z_{r'}} d(z,C)\geqslant d(a^*,C)=\max_{x\in P}d(x,C).
        \end{equation}
        \item  $r'\leqslant r_{min}$. In this case, recall that $Z_r \subseteq Y_r$ and that for every scale $r'\leqslant r_{min}$ and $\forall y\in Y_r$, $L_{y,r}=\{y\}$.
        Then
        $$Z_{r'/2}=\bigcup_{z\in Z_{r'}}\{y\in L_{z,r'}: d(y,C)\geqslant \max_{z\in Z_{r'}}d(z,C)-r'\}
        \subseteq
        \bigcup_{z\in Z_{r'}} \{z\} = Z_{r'}.
        $$
        
        
          
                    
     
        \end{enumerate}
        
        Now let 
        $r_1$ be the largest  scale for which
        $r_1 \leqslant \frac{1}{2}(\epsilon\cdot \max_{z\in Z_{r_1}} d(z,C))$ and $r_2$ the scale at which AFN($C,\epsilon$) terminates.
        
        From point 1, Equation (\ref{eq:point 1}) holds with $r'=r_1.$
        
        If $r_1 \ge r_{min}$, then $r_1=r_2$ and the lemma is correct.
        
        If $r_1 < r_{min}$ then $r_1 \le r_2 \le r_{min}$, so from point 2,  $Z_{r_1} \subseteq Z_{r_2}$ and
         $$
         (1+\epsilon)\cdot \max_{z\in Z_{r_2}} d(z,C)\geqslant
         (1+\epsilon)\cdot \max_{z\in Z_{r_1}} d(z,C)\geqslant d(a^*,C)=\max_{x\in P}d(x,C)$$
        Since $r_1$ satisfies condition 1, the second inequality holds. Hence, the lemma is again correct.
     %
        %
        \end{proof}
        
        We now analyze the running time of
        $\mbox{AFN}(P,C,\epsilon)$.
        
        \begin{lemma}\label{lem:zbound}
        In each iteration of
        $\mbox{AFN}(P,C,\epsilon)$,
        ${|Z_{r}|}\leqslant 4|C|(\gamma+2/\epsilon)^{O(\dim(\mathcal{X}))}$. \label{the size of $|Z_r|$}
        \end{lemma}
        \begin{proof}
        We actually prove the equivalent statement that $|Z_{r/2}|\leqslant 4|C|(\gamma+2/\epsilon)^{O(D)}$.

        For all $y\in Z_{r/2}$, there exists a point $z'\in Z_r$ satisfying $y\in L_{z',r}$, i.e, $d(z',y)\leqslant \gamma \cdot r$. Let $c'=\arg \min_{c\in C}d(z',C)$. Thus,
        $$d(y,c')\leqslant d(c',z')+d(z',y)=d(z',C)+d(z',y)\leqslant
        \max_{z\in Z_{r}}\,  d(z,C)+\gamma \cdot r.$$

        An iteration of  $\mbox{AFN}(C,\epsilon)$ will  construct $Z_{r/2}$ only when $\max_{z\in Z_r}d(z,C)\leqslant \frac{2r}{\epsilon}$. Therefore, $d(y,c')\leqslant(\gamma+2/\epsilon)r.$ This implies
        $Z_{r/2}\subseteq \bigcup_{c\in C} B(c,(\gamma+2/\epsilon)r).$
        
        Next notice that,  since $Z_{r/2}\subseteq Y_{r/2}$ is a $r/2$-net, $\forall z_1,z_2\in Z_{r/2},\, d(z_1,z_2)\geqslant\frac{r}{2}$.
        
        Finally, for fixed  $c\in C$, $\forall x,y\in Z_{r/2}\cap B(c,(\gamma+2/\epsilon)r)$, we have $\frac{r}{2}\leqslant d(x,y)\leqslant 2(\gamma+2/\epsilon)r$. Thus, the aspect ratio $\Delta_{B\left(c,(\gamma+2/\epsilon)r\right)}$ of the set $Z_{r/2}\cap B(c,(\gamma+2/\epsilon)r)$  is at most 
        $\Delta_{B(c,(\gamma+2/\epsilon)r)}\leqslant \frac{2(\gamma+2/\epsilon)r}{\frac{r}{2}}=4(\gamma+2/\epsilon)$.
        Therefore, by Lemma \ref{lem:doubling}, 
        $\forall c \in C,\, |Z_{r/2}\cap B(c,(\gamma+2/\epsilon)r)|\leqslant (4(\gamma+2/\epsilon))^{O(\dim(\mathcal{X}))}.$ 
        
        Thus, $|Z_{r/2}|\leqslant |C|(4(\gamma+2/\epsilon))^{O(\dim(\mathcal{X}))}$
        \end{proof}
        
        \begin{lemma}\label{lem:itertions}
        $\mbox{AFN}(P,C,\epsilon)$ runs at most 
         $\log_2 \Delta+O(1)$ iterations.
        \end{lemma}
        \begin{proof}
        The algorithm starts with $r=r_{max}$ and concludes when $r \ge r_{min}/2.$ Thus, the total number of iterations is at most
        $$
        \log_2 \frac {r_{max}} {r_{min}/2}=
        1+ \log_2 \frac {r_{max}} {r_{min}}
        = 1 + \log_2 \Theta\left(\frac {r_{max}} {r_{min}}\right)
        = 1 + \log_2 \Theta\left(\frac {d_{max}} {d_{min}}\right)=
        O(1) + \log_2 \Delta.
        $$
        \end{proof}
        
        Lemmas \ref{the size of $|Z_r|$} and  \ref{lem:itertions}, immediately imply that the  running time of AFN($C,\epsilon$) is at most $O(|C|(4(\gamma+2/\epsilon)))^{O(\dim(\mathcal{X}))}\log \Delta)$. 
        
        A  more careful analysis leads to the proof of \Cref{thm:NNAFN}. Due to space limitations the full proof  is deferred to \cref{App:careful_proof_of_running_time}.
%

\section{Modified $k$-Center Algorithms}\label{sec:alg}
    $AFN(P,C,\epsilon)$ will now  be used to design 
    two new 
    dynamic $k$-center algorithms. 
    
    Lemma \ref{lem:property_of_nn} hints that elements in $Y_r$ can be approximate {\em centers}.    
        This observation motivated Goranci et al.~\cite{goranci2021fully} to search  for the smallest $r$ such that $|Y_r|\leqslant k$ and return  the elements in $Y_r$ as centers. Unfortunately, used this way, the original navigating nets data structure 
        only returns an $8$-approximation solution. 
        \cite{goranci2021fully} improve  this by simultaneously maintaining multiple nets.     
        
%
        Although we also apply navigating nets to construct approximate  $k$-centers, our approach is very different from that of  \cite{goranci2021fully}. We do not use the elements in $Y_r$ as centers themselves. We only use the
        navigating net to support $AFN(P,C,\epsilon)$. Our algorithms result from substituting $AFN(P,c,\epsilon)$  for deterministic furthest  neighbor procedures in static algorithms.
 
    
    The next two subsections 
  introduce the two modified algorithms.

    \subsection{A Modified Version of  Gonzalez's \cite{gonzalez1985clustering}'s Greedy Algorithm}
Gonzalez
\cite{gonzalez1985clustering} described a simple and now well-known $O(kn)$ time $2$-approximation algorithm that works for any metric space. It operates by performing $k$ exact furthest neighbor from a set queries.
We just directly replace those exact queries with our new approximate furthest neighbor query procedure.

It is then straightforward  to modify Gonzalaz's proof  from \cite{gonzalez1985clustering}
 that his original algorithm is a $2$-approximation one, to prove that our new algorithm is a 
 $(2+\epsilon)$-approximation one.
 The details of the algorithm (\Cref{alg:GREEDY}) and the modified proof are provided in 
 \Cref{App:Gon}. This yields.
\begin{theorem}\label{correctness_of_greedy_algorithm}
    Let $P \subset \mathcal{X}$ be a finite set of points in a metric space $(\mathcal{X},d)$. Suppose $\mbox{AFN}(P,C,\epsilon)$ can be implemented in $T(|C|,\epsilon)$ time. \Cref{alg:GREEDY} constructs  a $(2+\epsilon)$-approximate solution for the $k$-center problem in 
       $O\left(k \cdot T\left(k, \frac {\epsilon} {5}\right)\right)$ time.
	\end{theorem}	
 Plugging \Cref{thm:NNAFN} into this proves   \Cref{thm:main2}.
 
\subsection{A Modified Version of the Kim  Schwarzwald \cite{kim1}  Algorithm} \label{Sec:KSmodified}
   
%
%
 
%

    In what follows, $D\ge 1$ is some arbitrary dimension.
    

	In 2020 
 \cite{kim1} gave  an $O(nD/\epsilon)$ time (1+$\epsilon$)-algorithm
	for the  Euclidean $1$-center (MEB) problem.  
	They further showed how to extend this 
 to obtain a (1+$\epsilon$)-approximation 
	to the Euclidean $k$-center in {$O(nD2^{O(k\log k/\epsilon)})$} time.

	Their algorithms use, as a subroutine,  a $\Theta(n)$ (or $\Theta(n |C|)$) time brute-force procedure for finding $\mbox{FN}(P,q)$ (or $\mbox{FN}(P,C)).$ 
	
	This subsection shows how replacing 
	$\mbox{FN}(P,q)$ (or $\mbox{FN}(P,C)$) by $\mbox{AFN}(P,q,\epsilon/3)$
	(or $\mbox{AFN}(P,C,\epsilon/3)$)
along with some other minor small changes, maintains the correctness of the algorithm. 
Our modified version of 
Kim and Schwarzwald \cite{kim1}'s
	 MEB  algorithm is presented as 
	Algorithm \ref{ALg:MEB}.


	Let $\epsilon>0$ be a constant.
    Their algorithm 
    runs in $O(1/\epsilon)$ iterations. The $i$'th  iteration  starts from some  point $m_i$ and uses $O(n)$ time to search for the point $p_{i+1}=\mbox{FN}(P,m_i)$ furthest from $m_i.$
    The  iteration  then  selects a ``good'' point $m_{i+1}$ on the line segment $p_{i+1} m_i$ as the starting point for the next iteration, where ``good''  means that the distance from $m_{i+1}$ to the optimal center is somehow bounded. The time to select such a  "good" point is $O(D)$.   The total running time of their algorithm is $O(nD/\epsilon)$. They also prove that the performance ratio of their algorithm is at most $(1+\epsilon)$.

	The running time of their algorithm is dominated by  the $O(n)$ time required to  find the point $\mbox{FN}(P,m_i)$.
	 As we will see in 
	Theorem \ref{correctness_of_algorithm} below,  
	finding the exact furthest point 
	$\mbox{FN}(P,m_i)$ was not necessary.  This could be replaced by $\mbox{AFN}(P,\epsilon/3,m_i).$ 
	%
	
	

	
	The first result is that this minor modification of  Kim and Schwarzwald \cite{kim1}'s algorithm still produces a $(1+\epsilon)$ approximation.

	
\begin{theorem}\label{correctness_of_algorithm}
	Let $P \subset \mathbb{R}^D$ be a set of points 
	whose minimum enclosing ball has (unknown) radius $r^*.$
	 Suppose $\mbox{AFN}(P,q,\epsilon)$ can be implemented in $T(\epsilon)$ time.

Let $c,r$ be the values returned by \Cref{ALg:MEB}. Then $P \subset B(c,r)$ and
$r \le (1+\epsilon) r^*$.  Thus \Cref{ALg:MEB}
 constructs  a $(1+\epsilon)$-approximate solution
and it runs in $O\left(DT\left(\frac \epsilon 3\right)\frac{1}{\epsilon}\right)$ time.
	\end{theorem}

 Plugging \Cref{thm:NNAFN} into 
 \Cref{correctness_of_algorithm} proves  \Cref{thm:main} for $k=1$.
	\begin{algorithm}[h!]
	\caption{Modified MEB($P,\epsilon$)} 
	\label{ALg:MEB}
	{\bf Input:} 
	A set of points $P$ and a constant $\epsilon>0$.\\
	{\bf Output:} A $(1+\epsilon)$-approximate minimum enclosing ball $B(c,r)$ containing all points in $P$.\\
    The algorithm presented is just a slight modification of that of \cite{kim1}. The differences are that in
	\cite{kim1},
	line 4 was originally 
	$p_{i+1}= \mbox{FN}(P,m_i)$ and the four  $(1 + \epsilon/3)$ terms on lines 8 and 9, were all originally $(1 + \epsilon).$
	
	\begin{algorithmic}[1]
	    \State Arbitrarily select a point $p_1$ from $P$;
	    \State Set $m_1=p_1$, $r=\infty$, and $\delta_1=1$;
	    \For{$i=1$ to $\lfloor 6/\epsilon \rfloor$}
	    \State $p_{i+1}=$AFN$(P,m_i,\epsilon/3)$;
	    \State $r_i=\left(1 + \frac \epsilon 3\right)d(m_i,p_{i+1});$
	     \If{$r_i < r$ }
	    \State $c=m_i;$ $r= r_i;$
	    \EndIf
	    \State $m_{i+1}=m_i+(p_{i+1}-m_i)\cdot \frac{\delta_i^2+(1+\epsilon/3)^2-1}
	    {2(1+\epsilon/3)^2}$
	    \State $\delta_{i+1}=\sqrt{1-\left(\frac{1+(1+\epsilon/3)^2-\delta^2_i}{2(1+\epsilon/3)}\right)^2}$;
	    \EndFor
	\end{algorithmic}
    \end{algorithm}	


\begin{proof}

Every ball 
$B(m_i,r_i)$
generated by 
\Cref{ALg:MEB} encloses all of the points in $P,$ i.e.,
\begin{equation}
\label{eq:Algoneball}
\forall i,\quad \max_{p\in P}d(m_i,p)\leqslant r_i. 
\end{equation}

To prove the correctness of the algorithm it suffices to show that $r \le (1 +\epsilon) r^*.$
Without loss of generality, we assume that $\epsilon \leqslant 1.$ 

    Each iteration of lines 4-9 of $\mbox{MEB}(P,\epsilon)$ must end in one of the two following cases: 
    \begin{enumerate}[(1)]
        \item $d(m_i,p_{i+1})\leqslant (1+\epsilon/3)r^*$,
        \item $d(m_i,p_{i+1})> (1+\epsilon/3)r^*$.
    \end{enumerate}

    Note that if Case (1) holds for some $i,$  then, directly from \Cref{eq:Algoneball}
    %
 (using $\epsilon \leqslant 1$),
    $$\max_{p\in P}d(m_i,p)\leqslant r_i =
    (1+\epsilon/3)d(m_i,p_{i+1}) \leqslant (1+\epsilon/3)^2 r^*<(1+\epsilon)r^*$$
    This implies that if  Case 1 ever holds, \Cref{ALg:MEB}  is correct.
    
    The main lemma is
     \begin{lemma}\label{lem:algonemainlemma}
    If,\  $ \forall 1\leqslant i\leqslant j$, case (2) holds, i.e., $d(m_i,p_{i+1})> (1+\epsilon/3)r^*$, then
    $j \le \frac{6}{\epsilon}-1.$
      \label{beta}
    \end{lemma}
    The proof of  \Cref {lem:algonemainlemma} is just a straightforward modification of the proof given in Kim and Schwarzwald \cite{kim1} for their original algorithm and is  therefore omitted. For completeness we provide the full modified proof in 
    \Cref{App:KSproof}.
    
    \Cref {lem:algonemainlemma} implies that, by the end of the algorithm, Case 1 must have occurred at least once, so
    $r \le (1+\epsilon)r^*$ and the algorithm outputs a correct solution. Derivation of the running time of the algorithm is straightforward, completing the proof of
    \Cref {correctness_of_algorithm}.
    \end{proof}
   
    \cite{kim1} discuss (without providing details) how to use the "guessing" technique of \cite{badoiu2002approximate,badoiu2003smaller}) to extend their MEB algorithm to yield  a  $(1+\epsilon)$-approximation solution to the $k$-center problem for $k\geqslant 2$.

    For MEB, the Euclidean  $1$-center, in each iteration, they  maintained the location of a candidate center $c$ and computed a furthest point to $c$ among $P$. For the Euclidean $k$-centers, in each step, they  maintain locations of a set $C$ of candidate centers,   $|C|\leqslant k$ and compute a furthest point to $C$ among $P$ using a $\mbox{FN}(P,C)$ procedure.
    
    Again we can modify their algorithm by replacing the $\mbox{FN}(P,C)$ procedure by a $\mbox{AFN}(P,C,\epsilon)$ one,  computing an {\em approximate}  furthest point to $C$ among $P$. This will prove \Cref{thm:main}.
    
    

    
    
    The full details of a modified version of their algorithm have been provided in \cref{App:KimAlgo}, which uses
	$\mbox{AFN}(P,C,\epsilon)$ in place of $\mbox{FN}(P,C)$, as well as an analysis of correctness and run time.

\section{Conclusion}
\label{sec:conc}

Our main new  technical contribution is 
 an algorithm, $\mbox{AFN}(P,C,\epsilon)$
that finds a 
 $(1+\epsilon)$-approximate furthest point in $P$ to $C.$
This works on top of a navigating net data structure \cite{krauthgamer2004navigating} storing $P.$

The proofs of  \Cref{thm:main2,thm:main}
follow immediately by
maintaining  a navigating net  and 
plugging 
$\mbox{AFN}(P,C,\epsilon)$ into 
\Cref{correctness_of_greedy_algorithm,thm:algframe}, respectively.

These provide 
a fully dynamic and deterministic $(2+\epsilon)$-approximation algorithm for the $k$-center problem in a metric space with finite doubling dimension and a $(1+\epsilon)$-approximation algorithm for the Euclidean $k$-center problem, where $\epsilon,k$ are  parameters given at query time.


One limitation of our algorithm is that, because $\mbox{AFN}(P,C,\epsilon)$  is built on top of navigating nets,  it depends upon aspect ratio $\Delta$. This is the only dependence of the $k$-center algorithm on $\Delta.$
An interesting future direction  would be to develop algorithms for $\mbox{AFN}(P,C,\epsilon)$ in special metric spaces 
built on top of other structures that are independent  of  $\Delta.$  This would automatically lead to algorithms for approximate $k$-center that, in those spaces,  would also be independent of $\Delta.$

\newcommand{\etalchar}[1]{$^{#1}$}

\newpage
\appendix

\section{A  Navigating Nets Example}
       \begin{figure}[h!]
        \centering
        \includegraphics[scale=0.2]{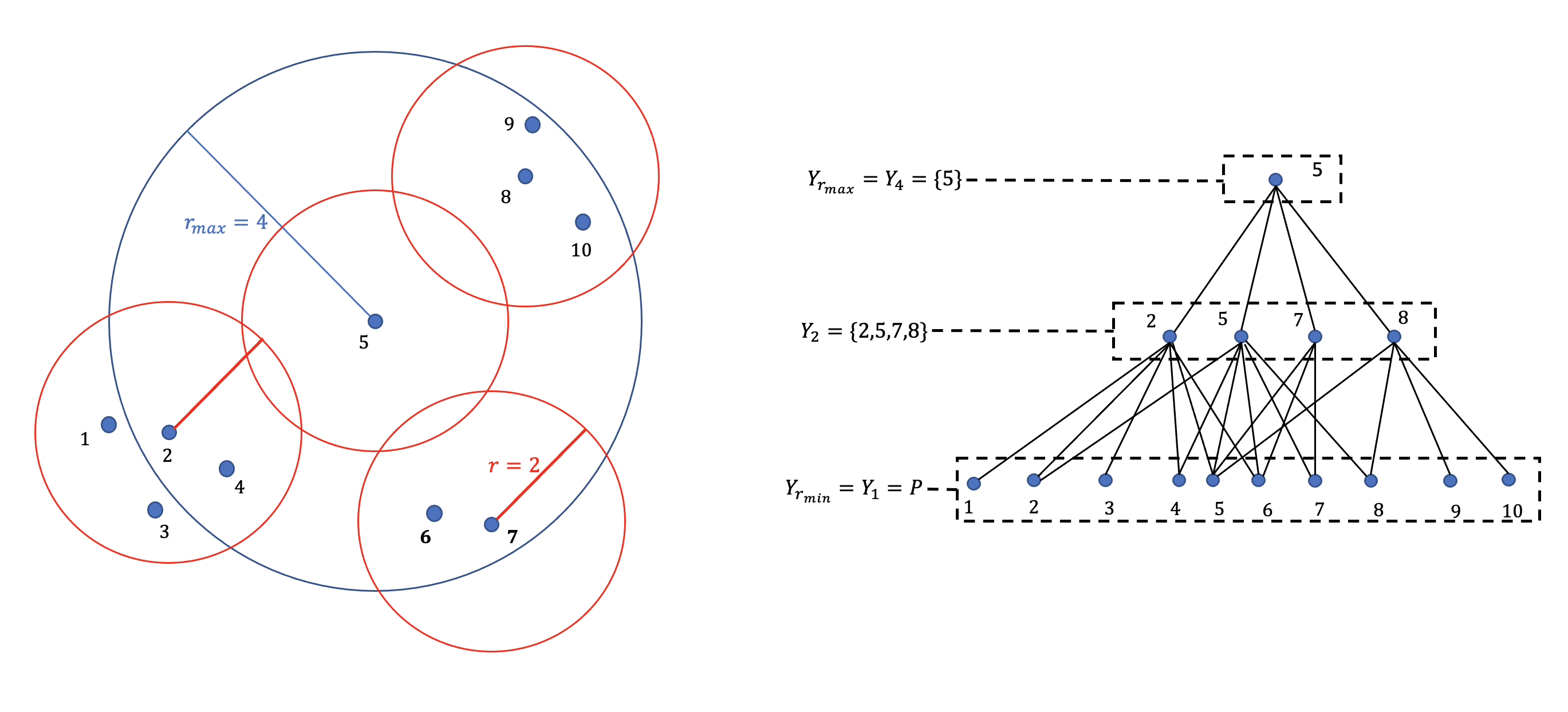}
        \caption{A navigating nets example with $P=\{1,2,...,10\},$ $r_{max}=4$ and $r_{min}=1.$ $Y_4=\{5\},$ $Y_2=\{2,5,7,8\}$ and $Y_{1}=P$. We note that $\forall x\in Y_2,\ d(x,Y_4)\leqslant 4$ and $\forall x\in P,\ d(x,Y_4)\leqslant 2\cdot 4$. $L_{2,2}=\{1,2,3,4,5,6\}$, $L_{5,2}=\{2,4,5,6,7,8\}$, $L_{7,2}=\{5,6,7\}$, and $L_{8,2}=\{5,8,9,10\}$. }
        \label{fig:NN}
        \end{figure}

\section{The Proof of \Cref {thm:NNAFN}}
\label{App:careful_proof_of_running_time}
       \begin{proof}

        
        It only remains to show that the  running time of $\mbox{AFN}(P,C,\epsilon)$ is bounded by $O(|C|\log \Delta) +O(|C|(1/\epsilon)^{O(\dim(\mathcal{X}))})$. We do this by splitting the set of scales $r$ processed by line 2 of \Cref{alg:AFNC2} into two ranges,
        (1) $\frac{r}{3}\geqslant \max_{x\in P}d(x,C)$ and (2) $\frac{r}{3}< \max_{x\in P}d(x,C).$  
        
        We then study the two cases separately : 
        For (1)  we will show a better bound on $|Z_r|$ than \Cref {lem:zbound}; for (2) we will show that the number of processed scales is small.
        %
        \begin{itemize}
            \item[(1)] When  $\frac{r}{3}\geqslant \max_{x\in P}d(x,C)$, the size of $Z_r$ is small. To see this note that 
            $$\max_{z\in Z_r}d(z,C)\leqslant\max_{x\in P}d(x,C)\leqslant \frac{r}{3}$$
            Thus, for each $z\in Z_r$, there exists a $c\in C$ such that $z\in B(c,\frac{r}{3})$, i.e., $Z_r\subseteq \bigcup_{c\in C}B(c,\frac{r}{3})$. Additionally, $Z_r\subseteq Y_r$, so, $\forall x,y\in Z_r$,  $d(x,y)\geqslant r$. Since the diameter of the ball $B(c,\frac{r}{3})$ is smaller than $r$, $|Z_r\cap B(c,\frac{r}{3})|\leqslant 1$ for every $c\in C$. Therefore, $|Z_r|\leqslant |C|$.
            
            
            By Lemma \ref{lem:itertions}, the number of iterations is at most $\log \Delta+O(1)$. Hence, the total running time for Case (1) is
            $O(|C|(\log \Delta +1))
            =O(|C|\log \Delta).$
            
            \item[(2)] When  $\frac{r}{3}< \max_{x\in P}d(x,C)$, although the size of $Z_r$ can be larger, the number of possible  iterations  will be small.
            
            By Lemma \ref{the size of $|Z_r|$}, $|Z_r|\leqslant 4|C|(\gamma+2/\epsilon)^{O(\dim(\mathcal{X}))}$. 
            Let $r'$ be the value when the algorithm terminates, then
            $2r'>\frac{\epsilon}{2} \max_{z\in Z_{r'}}d(z,C)$. 
            From Lemma \ref{lem:bound2},
            $(1+\epsilon)\max_{z\in Z_{r'}}d(z,C)\geqslant \max_{x\in P}d(x,C)$.We have
            $$3\cdot \max_{x\in P}d(x,C)>r\geqslant r'>\frac{\epsilon}{4}\max_{z\in Z_{r'}}d(x,C)\geq \frac{\epsilon}{4(1+\epsilon)}\max_{x\in P}d(x,C)$$
            

            Thus, the total number of Case (2) iterations  is at most $O(\log (\frac{1}{\epsilon}))$ and  the total  running time in  Case  (2) is $O(|C|(4(\gamma+2/\epsilon))^{O(\dim(\mathcal{X}))}\log (1/\epsilon)).$
        \end{itemize}
            Combining (1) and (2), the total running time of the algorithm is $O(|C|(\log \Delta))+O\left(|C|(1/\epsilon)^{O(\dim(\mathcal{X}))}\right).$
        \end{proof}

\section{The Modified Version of Gonzales's Greedy Algorithm}\label{App:Gon}

As noted in the main text, \Cref{alg:GREEDY} is essentially Gonzalez's \cite{gonzalez1985clustering}
 original algorithm  with $\mbox{FN}(P,C)$ replaced by   $\mbox{AFN}(P,C,\epsilon/5).$

\subsection{The Algorithm}   
    \begin{algorithm}[h!]
        \caption{Modified Greedy Algorithm: $GREEDY(P,\epsilon)$} 
        \label{alg:GREEDY}
        {\bf Input:} 
        A set of points $P\subset X$,  positive integer $k$  and a constant $\epsilon>0$.\\
        {\bf Output:} 
        A set $C$ ($|C| \le k$) and radius $r$ such that 
	 $P \subset \bigcup_{c \in C} B(c,r)$ and $r \le (2 + \epsilon) r^*.$
        \begin{algorithmic}[1]
            \State	Arbitrarily select a point $p_1$ from $P$ and set $C=\{p_1\}$.
            \While{$|C|< k$ } 
            \State $C=C\cup \left\{\mbox{AFN}\left(P,C,\frac{\epsilon}{5}\right)\right\}$ \quad\quad \quad \% $\mbox{FN}(P,C)$ is replaced by  $\mbox{AFN}(P,C,\epsilon/5)$.
            \EndWhile
            \State Set $r=\left(1+\frac{\epsilon}{5}\right)d(C,\mbox{AFN}(P,C,\frac{\epsilon}{5}))$
            \State Return $C,r$ as the solution.
        \end{algorithmic}
        \end{algorithm}	
     Gonzalaz's  original algorithm only returned $C,$ since in the deterministic case  $R = \mbox{FN}(P,C)$ could be calculated in $O(k r)$ time.

     \subsection{Proof of \Cref{correctness_of_greedy_algorithm}}

As noted this proof is just  a modification of the proof of correctness of 
 Gonzalez's \cite{gonzalez1985clustering}
 original algorithm (which used $\mbox{FN}(P,C)$ rather than $\mbox{AFN}(P,C,\epsilon/5)$).
 
    \begin{proof}
    (of \Cref{correctness_of_greedy_algorithm})

    Let $C=\{q_1,...,q_k\}$ and $r$ denote the solution returned by $GREEDY(P,\epsilon)$. 
    
    Let 
    $q=AFN(P,C,\frac{\epsilon}{5})$
    be the  
    $(1+\frac{\epsilon}{5})$-approximate
    furthest neighbor from $P$ to $C$ returned. Thus 
     $$\forall p\in P,\quad d(p,C)\leqslant \left(1+\frac{\epsilon}{5}\right)d(C,q)=r,$$ 
    i.e., $P\subseteq \bigcup_{i=1}^k B(c_i,r)$.  The output of the algorithm $GREEDY(P,\epsilon)$ is thus a feasible solution.
    
    Let $O=\{o_1,...,o_k\}$ denote an optimal $k$ center solution for the point set $P$ with  $r^*$ being the optimal radius value. Recall that $|S|$ is the number of points in set $S$. We consider two cases:
    
    \begin{enumerate}[\text{Case} 1]
        \item $\forall 1\leqslant i\leqslant k$, $|C\cap B(o_i,r^*)|=1$\\
        
        Fix $ p \in P.$  Let $o_i$ be such that 
        $p \in B(o_i,r^*).$
        Now let $q_j$ satisfy
        $q_j \in C \cap B(o_i,r^*)$.
        
        Then, by the triangle inequality, $d(p,q_j)\leqslant d(p,o_i)+d(o_i,q_j)\leqslant 2 r^*$.
        
        We have just shown that $\forall p\in P,$  $d(p,C)\leqslant 2r^*$.  In particular, 
       $$d(C,q)\leqslant \max_{p\in P} d(p,C)\leqslant 2r^*.$$ 
        Therefore, 
        $$r=\left(1+\frac{\epsilon}{5}\right)d(C,q)
        \leqslant
        \left(1+\frac{\epsilon}{5}\right) 2r^*
        < (2+\epsilon)r^*.$$

        \item There exists $o'\in O$ such that $|C\cap B(o',r^*)|\geqslant 2$.\\     
       Let  $q_i$ be  the $i$th point added into $C$ and  $C_i=\{q_1,...,q_i\}$. Thus $C_1\subset C_2\subset \cdots \subset C_k=C$. In Case 2, $C\cap B(o',r^*)$
       contains at least two points $q_i$ and $q_j$ ($i<j$). From line 3 in $GREEDY(P,\epsilon)$, we have  
       $q_j=AFN\left(P,C_{j-1},\frac{\epsilon}{5}\right)$.
       Furthermore, 
        $$
        \begin{aligned}
            \max_{p\in P} d(p,C)&\leqslant \max_{p\in P} d(p,C_{j-1}) \quad\quad \% \text{(Because $C_{j-1}\subseteq C$)}\\
                  &\leqslant 
                  \left(1+\frac{\epsilon}{5}\right)d(C_{j-1}, q_j)\\ 
                  & \leqslant \left(1+\frac{\epsilon}{5}\right)d(q_i,q_j)\quad\quad \% \text{(Because $q_i\in C_{j-1}$)}\\
                  & \leqslant \left(1+\frac{\epsilon}{5}\right) (d(q_i,o')+d(o',q_j)) \\
                  & \leqslant \left(1+\frac{\epsilon}{5}\right) (r^*+r^*)= \left(2+\frac{2\epsilon}{5}\right)r^*. \quad\quad  (*)
         \end{aligned}
        $$
        Then, consider the radius returned by $GREEDY(P,\epsilon)$:
        $$        
        \begin{aligned}
            r
            &=
             \left(1+\frac{\epsilon}{5}\right)d(C,q)\\
            &\leqslant 
            \left(1+\frac{\epsilon}{5}\right) \max_{p\in P} d(p,C)\\
            & \leqslant 
            \left(1+\frac{\epsilon}{5}\right)\left(2+\frac{2\epsilon}{5}\right)r^*   \quad\quad \% \text{(From ($*$))}\\
            & \leqslant
            \left( 2 + \frac {4 \epsilon}{5} + \frac {2 \epsilon^2} {25}\right) \leqslant  (2+\epsilon)r^*
        \end{aligned}
        $$
    \end{enumerate}
    where the last inequality assumes, without loss of generality, that $\epsilon \le 1.$
    
   Thus, $\forall p\in P$, $P\subseteq \bigcup_{i=1}^k B(c_i,R)$ and, in both cases $r\leqslant (2+\epsilon)r^*$. Thus, $GREEDY(P,\epsilon)$ always computes a $(2+\epsilon)$-approximation solution.
    \end{proof}

\section{Missing Details Associated with the  Modified Kim and Schwarzwald\cite{kim1}'s algorithm}

\subsection{Proof of \Cref {lem:algonemainlemma}}
\label{App:KSproof}

As noted previously, this proof is a modification of the proof  of correctness given by Kim and Schwarzwald \cite{kim1} for their algorithm (which used $\mbox{FN}(P,q)$ rather than $\mbox{AFN}(P,q,\epsilon/3)$). 

The proof needs  an important geometric observation due to 
   Kim and Schwarzwald\cite{kim1}  (slightly rephrased here). This was an extension of an earlier observation  by Kim and Ahn\cite{kim2015improved} that was used  to design a streaming algorithm for the Euclidean 2-center problem.

    \begin{lemma}\cite{kim1}\label{observation}
    (See Figure \ref{fig:observation}.) Fix $D\ge 2.$ Let  $B$ and $B'$
    be two $D$-dimensional balls with radii $r$ and $r',$ $r > r',$ around the same center point $c$. Let $p\in \partial B$ and $p'\in \partial B'$ with $d(p,p')=l\geqslant r$. 
    
    Define $B''$ to be the $D$-dimensional ball centered at  $c$ that is tangential to $pp'$. Denote that tangent  point as $m$ and define the distances $l_1=d(p',m)$  and $l_2=d(p,m)$.  Note that $l=l_1+l_2$. 
    
    Consider any line segment $p_1p_2$ satisfying $d(p_1,p_2)>l$, $p_1\in B'$ and $p_2\in B$. Then any point $m^*$ on $p_1p_2$ with $d(m^*,p_1)\geqslant l_1$ and $d(m^*,p_2)\geqslant l_2$ lies inside $B''$.
    \end{lemma}
    
    \begin{figure}[ht]
        \centering
        \includegraphics[scale=0.5]{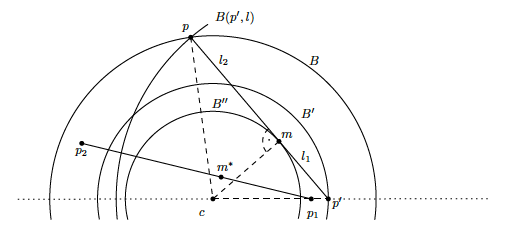}
        \caption{Illustration of lemma \ref{observation}. This figure is copied from  \cite{kim1}}
        \label{fig:observation}
    \end{figure}

    Lemma \ref{observation} will imply  that if 
     case (2) $d(m_i,p_{i+1})> (1+\epsilon/3)r^*$ always occurs, then  the distance from $c^*$ to $m_i$ is bounded. 
    
    \begin{corollary}
        If,\  $ \forall 1\leqslant i\leqslant j$, case (2) holds, i.e., $d(m_i,p_{i+1})> (1+\epsilon/3)r^*$, then  $d(c^*, m_{j+1})\leqslant \delta_{j+1}\cdot r^*$ and 
        $\delta_{j+1} <\delta_j.$
        \label{corollary}
    \end{corollary}
    
    \begin{proof}
        The proof will be by induction. In  the base case, $\delta_1=1$ and $m_1$ is arbitrarily selected in $P$, so $d(c^*, m_1)\leqslant \max_{x\in P}d(c^*,x)=r^*=\delta_1\cdot r^*$.

        Now suppose that,\  $ \forall 1\leqslant i\leqslant j$,  $d(m_i,p_{i+1})> (1+\epsilon/3)r^*$ holds.  From the induction hypothesis, we  assume that $d(c^*, m_j)\leqslant \delta_j\cdot r^*$ and
        $\delta_{j+1} < \delta_j \le 1.$
        
        
        Set $B=B(c^*,r^*)$, and $B'=B(c^*, \delta_j r^*)$.

        
        
        Next, arbitrarily select a point $p$ on the boundary of $B(c^*,r^*).$
        Now construct  ball $\Bar{B}=B(p,(1+\epsilon/3)r^*)$. From the induction hypothesis,
        $$d(m_j,p_{j+1})\leqslant d(c^*,m_j)+r^*\leqslant \delta_jr^*+r^*.$$
        If $\delta_j r^* \leq (\epsilon/3)r^*,$ this would imply
        $$d(m_j,p_{j+1})\leqslant(1+\epsilon/3)r^*,$$
        contradicting that this is Case 2.  Thus
        $\delta_j r^*> (\epsilon/3)r^* $.
        
        Since, $\delta_j < 1,$ this implies that  $\Bar{B}$ must intersect $B'$. Arbitrarily select one of the  two intersection points as $p'$.  Next set
     $l = d(p,p')=(1+\epsilon/3)r^*.$
     
     Finally, define $B''$ to be the ball centered at  $c^*$ that is tangent to line segment $pp'$. Denote that tangent  point as $m$ and define the distances $l_1=d(p',m)$  and $l_2=d(p,m)$. We will show that $B''=B(c^*,\delta_{j+1}r^*)$, i.e., that 
     $d(c^*,m)=\delta_{j+1}r^*$, where $\delta_{j+1}$ is as defined on line 9 of Algorithm \ref{ALg:MEB}.

        
        
        Note that $l=l_1+l_2$. Thus, $l_1$ can be computed by the equation (illustrated in figure \ref{fig:l_1}). 
        \begin{figure}[ht]
            \centering
            \includegraphics[scale=0.2]{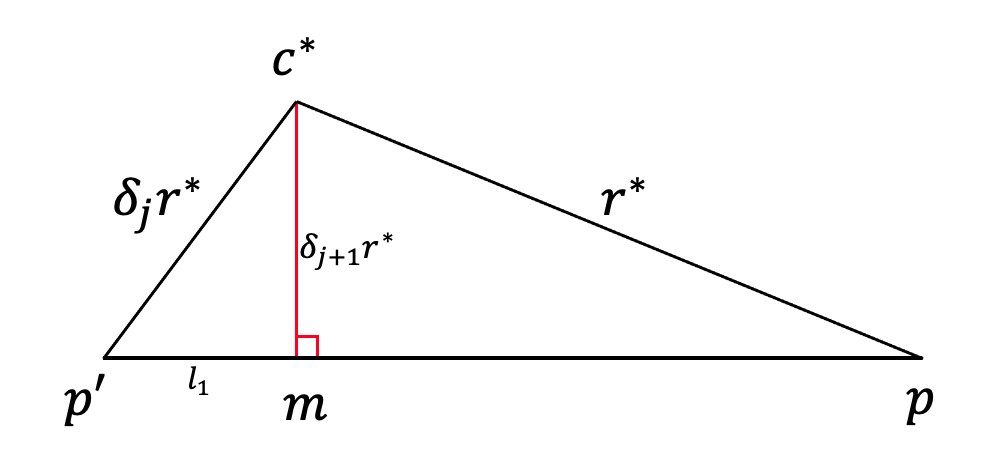}
            \caption{
            $l=d(p,p')=(1+\epsilon/3)r^*$, $d(p',m)=l_1$ and $d(m,p)=l-l_1$.}
            \label{fig:l_1}
        \end{figure}
        $$(\delta_j r^*)^2-l_1^2=(r^*)^2-((1+\epsilon/3)r^*-l_1)^2$$

        This solves to  $l_1= \frac{\delta_j^2+(1+\epsilon/3)^2-1}{2(1+\epsilon/3)}\cdot r^*.$ 
        Thus
        $$\delta_{j+1}=\sqrt{1-\left(\frac{1+(1+\epsilon/3)^2-\delta^2_j}{2(1+\epsilon/3)}\right)^2}$$
        as required and, by construction, $\delta_{j+1}< \delta_j.$  Furthermore,
        $$\frac{\delta_i^2+(1+\epsilon/3)^2-1}{2(1+\epsilon/3)^2} = \frac {l_1} l.$$
        Plugging  into line 8 of Algorithm \ref{ALg:MEB} yields  $$m_{j+1}=m_j+(p_{j+1}-m_j)\cdot \frac{l_1}{l}
        =m_j+(p_{j+1}-m_j)\cdot \left(1-\frac{l_2}{l}\right)
        .$$

        Since $d(m_j,p_{j+1})> (1+\epsilon/3)r^*$,  we have $$d(m_j,m_{j+1})=\frac{l_1}{l}\cdot d(m_j,p_{j+1})> \frac{l_1}{(1+\epsilon/3)r^*}(1+\epsilon/3)r^*=l_1$$ 
        and $d(p_{j+1},m_{j+1})=\frac{l_2}{l}\cdot d(m_j,p_{j+1})> l_2$. 
        
        From the induction hypothesis,  $d(c^*, m_j)\leqslant \delta_j\cdot r^*$, so $m_j\in B(c^*, \delta_j r^*)$. The  definition of $c^*,r^*,$  further implies $p_{j+1}\in B(c^*,r^*)$. 
        
        We now apply Lemma  \ref{observation}, with $p_1 = m_j$, $p_2 = p_{j+1}$ and $m^*=m_{j+1}.$ Since these three points are collinear, Lemma \ref{observation} implies that $m^* \in B'',$ i.e., that 
        $$d(c^*,m_{j+1})\leqslant \delta_{j+1}\cdot r^*.$$
    \end{proof}
    
    We can now 
    prove  \Cref{lem:algonemainlemma}. We again note that this is just a slight modification of the proof given by Kim and Schwarzwald \cite{kim1}
for their original algorithm. 

   





    \begin{proof} (of  \Cref{lem:algonemainlemma}.)

Recall that the goal is to prove that
if,\  $ \forall 1\leqslant i\leqslant j$, case (2) holds, i.e., $d(m_i,p_{i+1})> (1+\epsilon/3)r^*$, then
    $j \le \frac{6}{\epsilon}-1.$

    Consider the  triangle $\bigtriangleup p p' c^*$ 
    (Figure \ref{fig:l_1}) constructed in the proof of Corollary \ref{corollary}.
    
    Let $p(i), p'(i), m(i)$ denote $p,p',m$ in step $i$ of the algorithm.
    
    Recall that in the construction, $p(i)$ is on the boundary of $B(c^*,r^*)$ and $p'(i)$ is on the boundary of $B'=B(c^*,\delta_i r^*)$. Additionally, $d(p(i),p'(i))=(1+\epsilon/3)r^*$ and line segment $c^*m$ is vertical to line segment $p(i)p'(i).$  
    Recall that  $$d(p'(i),m(i))=l_1=\frac{\delta_i^2+(1+\epsilon/3)^2-1}{2(1+\epsilon/3)}r^*
    \quad\mbox{and}\quad d(p(i),m(i))=(1-\epsilon/3)r^*-l_1.$$ 
    
    Now define $\beta_i,\alpha_i$ so that
    $$d(m(i),p(i))=\beta_i \cdot d(p(i),p'(i))=\beta_i(1+\epsilon/3)r^*
    \quad\mbox{and}\quad(p'(i),m(i))=\alpha_i (1+\epsilon/3)r^*.$$ 
    Note that $\alpha_i + \beta_i=1.$
    

    Since $\bigtriangleup m(i)p(i) c^*$ is a right triangle, $d(m(i),p(i))=\beta_i(1+\epsilon/3)r^* \leqslant d(c^*,p(i))= r^*$ so, $\forall i,$   $\beta_i \le \frac{1}{(1+\epsilon/3)}$.

    We have therefore just proven that 
    if,\  $ \forall 1\leqslant i\leqslant j$, case (2) holds, i.e., $d(m_i,p_{i+1})> (1+\epsilon/3)r^*$, then for all such $i,$ 
     $\beta_i \le \frac{1}{(1+\epsilon/3)}$ and, in particular, 
     $\beta_j \le \frac{1}{(1+\epsilon/3)}.$
    
    


         
            By construction,  $(\delta_i r^*)^2=(r^*)^2-(\beta_{i}(1+\epsilon/3)r^*)^2$ and $(\alpha_{i}(1+\epsilon/3)r^*)^2=(\delta_{i-1} r^*)^2-(\delta_i r^*)^2$. 
            Plugging the first (twice) into the second yields
            $$(\alpha_{i}(1+\epsilon/3)r^*)^2=\left((r^*)^2-(\beta_{i-1}(1+\epsilon/3)r^*)^2\right)
            -\left((r^*)^2-(\beta_{i}(1+\epsilon/3)r^*)^2\right),
            $$
            or $\beta_i^2=\alpha_i^2+\beta_{i-1}^2.$
        
        Combining  this with
         $\beta_i=1-\alpha_i$
       yields 
         $\beta_i=\frac{1+\beta^2_{i-1}}{2}$. Set $\varphi_i=\frac{1}{1-\beta_i}$. Then
        $$\varphi_i=\frac{1}{1-\beta_i}=\frac{1}{1-\frac{1+\beta_{i-1}^2}{2}}=\frac{1}{\frac{1-\beta_{i-1}^2}{2}}=\frac{\frac{1}{1-\beta_{i-1}}}{\frac{1+\beta_{i-1}}{2}}=\frac{\varphi_{i-1}}{\frac{1+(1-\frac{1}{\varphi_{i-1}})}{2}}=\frac{\varphi_{i-1}}{1-\frac{1}{2\varphi_{i-1}}}.$$
        Thus
        $$\varphi_i=\frac{\varphi_{i-1}}{1-\frac{1}{2\varphi_{i-1}}}=\varphi_{i-1}(1+\frac{1}{2\varphi_{i-1}}+\frac{1}{(2\varphi_{i-1})^2}+\cdots)\geqslant \varphi_{i-1}+\frac{1}{2}.$$

       
      Recall that $\delta_1=1$, $d(c^*,p(1))=d(c^*,p'(1))=r^*$, i.e., $\bigtriangleup m(i)p(i) c^*$ is an isosceles triangle. Therefore, 
       $\alpha_1=\beta_1=\frac{1}{2},$  and  $\varphi_1=2.$ Iterating the equation above yields $\varphi_i\geqslant 2+\frac{i-1}{2}$. Thus, $\beta_i\geqslant 1-\frac{2}{3+i}$. Thus, if $j > \frac{6}{\epsilon}-1$, then $\beta_j> \frac{1}{1+\epsilon/3},$
        which we previously saw was not possible.
        
    \end{proof}

\subsection{The Actual Modified Algorithm for  Euclidean $k$ Center }\label{App:KimAlgo}
    	
	\begin{algorithm}[h!]
	\caption{Modified $k$-center($P,\epsilon,k$)} 
	\label{ALg:kcenter}
	{\bf Input:} 
	A set of points $P$, positive integer $k$ and a constant $\epsilon>0$.\\
	{\bf Output:} A set $\bar C$ ($|\bar C| \le k$) and radius $\bar r$ such that 
	 $P \subset \bigcup_{c \in \bar C} B(c,r)$ and $\bar r \le (1 + \epsilon) r^*.$\\ 
	In the algorithm, each $m_{j,i}$, $j \in \{1,\ldots,k\},$ is either  undefined or a point in $\mathbb{R}^D.$
	$M_i$ denotes the set of defined $m_{j,i}.$ 
	$\mathcal{F}$ is the set of all functions from 
	$\{1,\ldots, k\lfloor 6/\epsilon \rfloor\}$ to 
	$\{1,\ldots,k\}.$
	\begin{algorithmic}[1]
    \State{$\bar r=\infty$}
	\For{every function $f \in \mathcal{F}$}
	    \State {$\forall j\in\{1\ldots k\},$ set 
	    $\delta_{j,1}=1;$ Set $r=\infty;$}
	    \State Arbitrarily select a point $p_1$ from $P$;
	    \State $m_{j,1} = \begin{cases}
	    \mbox{undefined} & \mbox{if $j \not=f(1)$}\\
	    p_1 & \mbox{if $j= f(1)$}
	    \end{cases}$
	    \For{$i=1$ to $k\lfloor 6/\epsilon \rfloor$}
	    \State $p_{i+1}=$AFN$(P,M_i,\epsilon/3)$;
	    \State $r_i=\left(1 + \frac \epsilon 3\right)d(M_i,p_{i+1});$
	    \If{$r_i< r$ }
	    \State $C=M_i$;  $r=r_i$;
	    \EndIf
	       \State $m_{j,i+1} = \begin{cases}
	    m_{j,i} & \mbox{if $j \not=f(i)$}\\
	    m_{j,i}+(p_{i+1}-m_{j,i})\cdot \frac{\delta_{j,i}^2+(1+\epsilon/3)^2-1} {2(1+\epsilon/3)^2}& \mbox{if $j= f(i)$}
	    \end{cases}
	    $
	       \State $\delta_{j,i+1} = \begin{cases}
	    \delta_{j,i} & \mbox{if $j \not=f(i)$}\\
	    \sqrt{1-\left(\frac{1+(1+\epsilon/3)^2-\delta^2_{j,i}}{2(1+\epsilon/3)}\right)^2}& \mbox{if $j= f(i)$}
	    \end{cases}
	    $
	    \EndFor
	    \If{$ r <  \bar r$ }
	        \State{ $\bar r =r;$ $\bar C = C;$ }
	    \EndIf
	    \EndFor
	\end{algorithmic}
    \end{algorithm}

    Before starting, we provide a brief  intuition.  If, for each $p \in P$, the algorithm knew in advance  in which of the $k$ clusters $p$ is located, it could solve the problem by running \Cref {ALg:MEB}  separately for each cluster and returning the largest radius it found. Since it doesn't know that information in advance, it ``guesses'' the location.  This guess is encoded by the function $f \in {\mathcal F}$ introduced in \Cref{ALg:kcenter}. It runs this procedure for every possible guess.  Since one of the guesses must be correct, the algorithm  returns a correct answer.

    Again, we emphasize that \Cref{ALg:kcenter}
    is essentially the algorithm alluded\footnote{We write ``alluded to'' because \cite{kim1} do not actually provide details. They only say that they are utilizing the guessing technique from \cite{badoiu2003smaller}. In our algorithm, we have provided full details of how this can be done.  } to in Kim and Schwarzwald \cite{kim1} with calls to 
  $\mbox{FN}(P,C)$ replaced by calls to $\mbox{AFN}(P,C,\epsilon/3).$
    \begin{theorem}\label{thm:algframe}
    Let $P \subset \mathbb{R}^D$  be a finite set of points. Suppose $\mbox{AFN}(P,C,\epsilon)$ can be implemented in $T(|C|, \epsilon)$ time.
     Then an  $(1+\epsilon)$-approximate $k$-center solution for $P$ can be constructed in  $O\left(DT\left(k, \frac \epsilon 3\right)2^{O(k\log k/\epsilon)}\right)$ time.
    \end{theorem}
    \begin{proof}
     Let $C^*=\{c_1^*,\ldots,c^*_k\}$ be a set of optimal centers and $r^*= d(C^*,P).$ Partition the points in $P$ into $P_i,$ $i=1,\ldots,k,$ so that
     $P_i \subseteq P \cap B(c^*_i,r^*).$ Let $r^*_i=d(c^*_i,P_i),$ i.e., $B(c^*_i, r^*_i)$ is a minimum enclosing ball for $P_i.$  Note that 
     $r^* = \max_i r^*_i.$
     
     Fix $f: \{1,\ldots,k\lfloor 6/\epsilon \rfloor\} \rightarrow \{1,\ldots,k\}$ to be some arbitrary function.  Lines 3-12 of \Cref {ALg:kcenter}
     maintains a list
     of $k$ tentative centers; $m_{1,i},\ldots,m_{k,i},$  denotes the list  at the start of iteration $i.$
     Note that some of the $m_{j,i}$ might be undefined, i.e., do not exist. $M_i$ will denote the set of defined items in the list at the start of iteration $i.$ During iteration $i$, the list updates (only the) tentative center  $m_{f(i),i}$ and  also constructs a radius $r_i.$ 
     
     The algorithm starts with all of the $m_{j,i}$ being undefined, chooses some arbitrary point of $P,$ calls it $p_1,$ and then sets $m_{f(1),1}=p_1.$
     
     At step $i,$ it sets $p_{i+1}=\mbox{AFN}\left(P,M_i,\epsilon/3\right)$ and $r_i = \left(1+ \frac \epsilon 3 \right) d(M_i,p_{i+1}).$
     
     Note that the definitions of 
     $\mbox{AFN}\left(P,M_i,\epsilon/3\right)$ and $r_i$  immediately imply
     \begin{equation}
     \label{eq:alg1k}
         \forall i,\quad P \subset \bigcup_{q \in M_i} B(q,r_i).
     \end{equation}
     Thus, lines 3-12  return $C$ and $r$ that cover all  points in $P$ in $O\left( D \frac k \epsilon T\left(|C|,\frac \epsilon 3\right)\right)$ time.

     So far, the analysis has not considered  lines 11-12 of the algorithm.
     
    The algorithm arbitrarily chooses $p_1.$ Now consider the unique function $ f'(i)$  that  always returns the index of the $P_j$ that contains $p_{i+1},$ i.e., $p_{i+1} \in P_{f'(i)}.$ 
    For this $f'(i),$ lines 11 and 12 of the algorithm work as if they are running the original modified MEB algorithm on each of the $P_j$ separately.

     By the generalized pigeonhole principle, there must exist at least one index $j$ such that $f'(i)=j$ at least $\lfloor 6/\epsilon\rfloor$ times. For such a $j,$ consider the value of $i$ for which $f'(i)=j$ exactly $\lfloor 6/\epsilon\rfloor$ times. Then, from the analysis of 
     \Cref{ALg:MEB}, for this particular $j,$
     $$ d(m_{j,i},p_{i+1})\le (1 +  \epsilon/3) r^*_j,$$
   so
     $$d(M_i,p_{i+1})\le d(m_{j,i},p_{i+1})\le (1 +  \epsilon/3) r^*_j\le
     (1 +  \epsilon/3) r^*.$$
    
     Thus
     $$  r_i=\left(1 + \frac \epsilon 3\right)d(M,p_{i+1})\le
     \left(1 + \frac \epsilon 3\right)^2 r^*
     \le \left(1 + \epsilon \right) r^*.
     $$

In particular, lines 3-12 
using  $f'(i)$ returns a 
$(1+ \epsilon)$-approximate solution.

\Cref{ALg:kcenter} runs lines 3-12
on all 
$\Theta\left(
k^{k\lfloor 6/\epsilon \rfloor}
\right) = 2^{O(k \log k/\epsilon)}$
different possible functions $f(i).$ Since this includes $f'(i),$ the full algorithm also returns a $(1+ \epsilon)$-approximate solution.
    \end{proof}

\end{document}